\documentclass[letterpaper]{article}
\usepackage[latin9]{inputenc}
\usepackage{amsthm}
\usepackage{amsmath}
\usepackage{amssymb}
\usepackage{esint}

\makeatletter

\providecommand{\tabularnewline}{\\}

\theoremstyle{plain}
\newtheorem{thm}{\protect\theoremname}
\theoremstyle{definition}
\newtheorem{defn}[thm]{\protect\definitionname}
\theoremstyle{plain}
\newtheorem{prop}[thm]{\protect\propositionname}
\theoremstyle{plain}
\newtheorem{lem}[thm]{\protect\lemmaname}
\ifx\proof\undefined\
\newenvironment{proof}[1][\protect\proofname]{\par
\normalfont\topsep6\p@\@plus6\p@\relax
\trivlist
\itemindent\parindent
\item[\hskip\labelsep
\scshape
#1]\ignorespaces
}{%
\endtrivlist\@endpefalse
}
\fi
\theoremstyle{plain}
\newtheorem{cor}[thm]{\protect\corollaryname}

\usepackage{nips13submit_e,times}
\usepackage{url}
\usepackage{pstool}
\usepackage{subfig}

\author{
 Dino Sejdinovic, Arthur Gretton \\
Gatsby Unit, CSML, UCL, UK \\
\texttt{ \{dino.sejdinovic,arthur.gretton\}@gmail.com} \\
\AND
Wicher Bergsma \\
Department of Statistics, LSE, UK \\
\texttt{w.p.bergsma@lse.ac.uk}
}

\newcommand\ci{\protect\mathpalette{\protect\cI}{\perp}}
\def\cI#1#2{\mathrel{\rlap{$#1#2$}\mkern2mu{#1#2}}}

\def\cI#1#2{\mathrel{\rlap{$#1#2$}\mkern2mu{#1#2}}}

\nipsfinalcopy 

\makeatother

\providecommand{\corollaryname}{Corollary}
\providecommand{\definitionname}{Definition}
\providecommand{\lemmaname}{Lemma}
\providecommand{\proofname}{Proof}
\providecommand{\propositionname}{Proposition}
\providecommand{\theoremname}{Theorem}

\begin{document}
\global\long\def\one{\mathbf{1}}
\global\long\def\onet{\mathbf{1}^{\top}}
\global\long\def\sumall{\Xi}
\global\long\def\iid{\overset{i.i.d.}{\sim}}

\title{A Kernel Test for Three-Variable Interactions}
\maketitle
\begin{abstract}
We introduce kernel nonparametric tests for Lancaster three-variable
interaction and for total independence, using embeddings of signed
measures into a reproducing kernel Hilbert space. The resulting test
statistics are straightforward to compute, and are used in powerful
interaction tests, which are consistent against all alternatives for
a large family of reproducing kernels. We show the Lancaster test
to be sensitive to cases where two independent causes individually
have weak influence on a third dependent variable, but their combined
effect has a strong influence. This makes the Lancaster test especially
suited to finding structure in directed graphical models, where it
outperforms competing nonparametric tests in detecting such V-structures.
\end{abstract}

\section{Introduction}

The problem of nonparametric testing of interaction between variables
has been widely treated in the machine learning and statistics literature.
Much of the work in this area focuses on measuring or testing pairwise
interaction: for instance, the Hilbert-Schmidt Independence Criterion
(HSIC) or Distance Covariance \cite{GreBouSmoSch05,Szekely2007,Sejdinovic2012},
kernel canonical correlation \cite{BacJor02,FukBacGre07,DauNki98},
and mutual information \cite{pal10renyi}. In cases where more than
two variables interact, however, the questions we can ask about their
interaction become significantly more involved. The simplest case
we might consider is whether the variables are mutually independent,
$P_{X}=\prod_{i=1}^{d}P_{X_{i}}$, as considered in $\mathbb{R}^{d}$
by \cite{Kankainen95}. This is already a more general question than
pairwise independence, since pairwise independence does not imply
total (mutual) independence, while the implication holds in the other
direction. For example, if $X$ and $Y$ are i.i.d. uniform on $\left\{ -1,1\right\} $,
then $\left(X,Y,XY\right)$ is a pairwise independent but mutually
dependent triplet \cite{Bernstein46}. Tests of total and pairwise
independence are insufficient, however, since they do not rule out
all third order factorizations of the joint distribution. 

An important class of high order interactions occurs when the simultaneous
effect of two variables on a third may not be additive. In particular,
it may be possible that $X\ci Z$ and $Y\ci Z$, whereas $\neg\left((X,Y)\ci Z\right)$
(for example, neither adding sugar to coffee nor stirring the coffee
individually have an effect on its sweetness but the joint presence
of the two does). In addition, study of three-variable interactions
can elucidate certain switching mechanisms between positive and negative
correlation of two genes expressions, as controlled by a third gene
\cite{Kayano2009}. The presence of such interactions is typically
tested using some form of analysis of variance (ANOVA) model which
includes additional interaction terms, such as products of individual
variables. Since each such additional term requires a new hypothesis
test, this increases the risk that some hypothesis test will produce
a false positive by chance. Therefore, a test that is able to directly
detect the presence of \emph{any kind }of\emph{ }higher-order interaction
would be of a broad interest in statistical modeling. In the present
work, we provide to our knowledge the first nonparametric test for
three-variable interaction. This work generalizes the HSIC test of
pairwise independence, and has as its test statistic the norm of an
embedding of an appropriate signed measure to a reproducing kernel
Hilbert space (RKHS). When the statistic is non-zero, all third order
factorizations can be ruled out. Moreover, this test is applicable
to the cases where $X$, $Y$ and $Z$ are themselves multivariate
objects, and may take values in non-Euclidean or structured domains.%
\footnote{As the reader might imagine, the situation becomes more complex again
when four or more variables interact simultaneously; we provide a
brief technical overview in Section \ref{sec:The-case-D-bigger-3}.%
}

One important application of interaction measures is in learning structure
for graphical models. If the graphical model is assumed to be Gaussian,
then second order interaction statistics may be used to construct
an undirected graph \cite{MeiBuh06,RavWaiRasYu11}. When the interactions
are non-Gaussian, however, other approaches are brought to bear. An
alternative approach to structure learning is to employ conditional
independence tests. In the PC algorithm \cite{Pearl01,spirtes:00,KalBuh07},
a V-structure (two independent variables with directed edges towards
a third variable) is detected when an independence test between the
parent variables accepts the null hypothesis, while a test of dependence
of the parents conditioned on the child rejects the null hypothesis.
The PC algorithm gives a correct equivalence class of structures subject
to the causal Markov and faithfulness assumptions, in the absence
of hidden common causes. The original implementations of the PC algorithm
rely on partial correlations for testing, and assume Gaussianity.
A number of algorithms have since extended the basic PC algorithm
to arbitrary probability distributions over multivariate random variables
\cite{SunJanSchFuk07,TilGreSpi09,Zhang2011}, by using nonparametric
kernel independence tests \cite{GreFukTeoSonetal08} and conditional
dependence tests \cite{FukGreSunSch08,Zhang2011}.  We observe that
our Lancaster interaction based test provides a strong alternative
to the conditional dependence testing approach, and is seen to outperform
earlier approaches in detecting cases where independent parent variables
weakly influence the child variable when considered individually,
but have a strong combined influence. 

We begin our presentation in Section \ref{sub:Interaction-measure}
with a definition of interaction measures, these being the signed
measures we will embed in an RKHS. We cover this embedding procedure
in Section \ref{sec:Kernel-Embeddings}. We then proceed in Section
\ref{sec:Interaction-tests} to define pairwise and three way interactions.
We describe a statistic to test mutual independence for more than
three variables, and provide a brief overview of the more complex
high-order interactions that may be observed when four or more variables
are considered. Finally, we provide experimental benchmarks in Section
\ref{sec:Experiments}.

Matlab code for interaction tests considered in the paper is available at \url{http://www.gatsby.ucl.ac.uk/~gretton/interact/threeWayInteract.htm}

\section{Interaction measure\label{sub:Interaction-measure}}

An interaction measure \cite{Lancaster1969,Streitberg1990} associated
to a multidimensional probability distribution $P$ of a random vector
$\left(X_{1},\ldots,X_{D}\right)$ taking values in the product space
$\mathcal{X}_{1}\times\cdots\times\mathcal{X}_{D}$ is a signed measure
$\Delta P$ that vanishes whenever $P$ can be factorised in a non-trivial
way as a product of its (possibly multivariate) marginal distributions.
For the cases $D=2,3$ the correct interaction measure coincides with
the the notion introduced by Lancaster \cite{Lancaster1969} as a
formal product
\begin{eqnarray}
\Delta_{L}P & = & \prod_{i=1}^{D}\left(P_{X_{i}}^{*}-P_{X_{i}}\right),\label{eq: Lancaster}
\end{eqnarray}
where $\prod_{j=1}^{D'}P_{X_{i_{j}}}^{*}$ is understood as a joint
probability distribution of a subvector $\left(X_{i_{1}},\ldots,X_{i_{D'}}\right)$.
We will term the signed measure in \eqref{eq: Lancaster} the \emph{Lancaster
interaction measure}. In the case of a bivariate distribution, the
Lancaster interaction measure is simply the difference between the
joint probability distribution and the product of the marginal distributions
(the only possible non-trivial factorization for $D=2$), $\Delta_{L}P=P_{XY}-P_{X}P_{Y}$,
while in the case $D=3$, we obtain
\begin{eqnarray}
\Delta_{L}P & = & P_{XYZ}-P_{XY}P_{Z}-P_{YZ}P_{X}-P_{XZ}P_{Y}+2P_{X}P_{Y}P_{Z}.\label{eq: 3Lancaster}
\end{eqnarray}
It is readily checked that
\begin{eqnarray}
(X,Y)\ci Z\,\vee\,(X,Z)\ci Y\,\vee\,(Y,Z)\ci X\, & \Rightarrow & \Delta_{L}P=0.\label{eq: H0-Lancaster}
\end{eqnarray}
For $D>3$, however, \eqref{eq: Lancaster} does not capture all possible
factorizations of the joint distribution, e.g., for $D=4$, it need
not vanish if $(X_{1},X_{2})\ci(X_{3},X_{4})$, but $X_{1}$ and $X_{2}$
are dependent and $X_{3}$ and $X_{4}$ are dependent. Streitberg
\cite{Streitberg1990} corrected this definition using a more complicated
construction with the M\"{o}bius function on the lattice of partitions,
which we describe in Section \ref{sec:The-case-D-bigger-3}. In this
work, however, we will focus on the case of three variables and formulate
interaction tests based on embedding of \eqref{eq: 3Lancaster} into
an RKHS.

The implication \eqref{eq: H0-Lancaster} states that the presence
of Lancaster interaction rules out the possibility of any factorization
of the joint distribution, but the converse is not generally true;
see Appendix \ref{sec: counterexample} for details. In addition,
it is important to note the distinction between the absence of Lancaster
interaction and the total (mutual) independence of $(X,Y,Z)$, i.e.,
$P_{XYZ}=P_{X}P_{Y}P_{Z}$. While total independence implies the absence
of Lancaster interaction, the signed measure $\Delta_{tot}P=P_{XYZ}-P_{X}P_{Y}P_{Z}$
associated to the total (mutual) independence of $(X,Y,Z)$ does not
vanish if, e.g., $(X,Y)\ci Z$, but $X$ and $Y$ are dependent. 

In this contribution, we construct the non-parametric test for the
hypothesis $\Delta_{L}P=0$ (no Lancaster interaction), as well as
the non-parametric test for the hypothesis $\Delta_{tot}P=0$ (total
independence), based on the embeddings of the corresponding signed
measures $\Delta_{L}P$ and $\Delta_{tot}P$ into an RKHS. Both tests
are particularly suited to the cases where $X$, $Y$ and $Z$ take
values in a high-dimensional space, and, moreover, they remain valid
for a variety of non-Euclidean and structured domains, i.e., for all
topological spaces where it is possible to construct a valid positive
definite function; see \cite{FukSriGreSch09} for details. In the
case of total independence testing, our approach can be viewed as
a generalization of the tests proposed in \cite{Kankainen1995} based
on the empirical characteristic functions.

\section{Kernel Embeddings \label{sec:Kernel-Embeddings}}

We review the embedding of signed measures to a reproducing kernel
Hilbert space. The RKHS norms of such embeddings will then serve as
our test statistics. Let $\mathcal{Z}$ be a topological space. According
to the Moore-Aronszajn theorem \cite[p. 19]{BerTho04}, for every
symmetric, positive definite function (henceforth \emph{kernel}) $k:\mathcal{Z}\times\mathcal{Z}\to\mathbb{R}$,
there is an associated reproducing kernel Hilbert space (RKHS) $\mathcal{H}_{k}$
of real-valued functions on $\mathcal{Z}$ with reproducing kernel
$k$. The map $\varphi:\mathcal{Z}\to\mathcal{H}_{k}$, $\varphi:z\mapsto k(\cdot,z)$
is called the canonical feature map or the Aronszajn map of $k$.
Denote by $\mathcal{M}(\mathcal{Z})$ the Banach space of all finite
signed Borel measures on $\mathcal{Z}$. The notion of a feature map
can then be extended to kernel embeddings of elements of $\mathcal{M}(\mathcal{Z})$
\cite[Chapter 4]{BerTho04}. 
\begin{defn}
(\textbf{Kernel embedding}) Let $k$ be a kernel on $\mathcal{Z}$,
and $\nu\in\mathcal{M}(\mathcal{Z})$. The \emph{kernel embedding}
of $\nu$ into the RKHS $\mathcal{H}_{k}$ is $\mu_{k}(\nu)\in\mathcal{H}_{k}$
such that $\int f(z)d\nu(z)=\left\langle f,\mu_{k}(\nu)\right\rangle _{\mathcal{H}_{k}}$
for all $f\in\mathcal{H}_{k}$.
\end{defn}
Alternatively, the kernel embedding can be defined by the Bochner
integral $\mu_{k}(\nu)=\int k(\cdot,z)\, d\nu(z)$. If a measurable
kernel $k$ is a bounded function, it is straightforward to show using
the Riesz representation theorem that $\mu_{k}(\nu)$ exists for all
$\nu\in\mathcal{M}(\mathcal{Z})$.%
\footnote{Unbounded kernels can also be considered, however \cite{Sejdinovic2012a}.
In this case, one can still study embeddings of the signed measures
$\mathcal{M}_{k}^{1/2}(\mathcal{Z})\subset\mathcal{M}(\mathcal{Z})$,
which satisfy a finite moment condition, i.e., $\mathcal{M}_{k}^{1/2}(\mathcal{Z})=\left\{ \nu\in\mathcal{M}(\mathcal{Z})\,:\,\int k^{1/2}(z,z)\, d\vert\nu\vert(z)<\infty\right\} $
.%
} For many interesting bounded kernels $k$, including the Gaussian,
Laplacian and inverse multiquadratics, the embedding $\mu_{k}:\mathcal{M}(\mathcal{Z})\to\mathcal{H}_{k}$
is injective. Such kernels are said to be \emph{integrally strictly
positive definite} (ISPD) \cite[p. 4]{Sriperumbudur2011}. A related
but weaker notion is that of a \emph{characteristic} kernel \cite{FukGreSunSch08,SriGreFukLanetal10},
which requires the kernel embedding to be injective only on the set
$\mathcal{M}_{+}^{1}(\mathcal{Z})$ of probability measures. In the
case that $k$ is ISPD, since $\mathcal{H}_{k}$ is a Hilbert space,
we can introduce a notion of an inner product between two signed measures
$\nu,\nu'\in\mathcal{M}(\mathcal{Z})$,
\[
\left\langle \left\langle \nu,\nu'\right\rangle \right\rangle _{k}:=\left\langle \mu_{k}(\nu),\mu_{k}(\nu')\right\rangle _{\mathcal{H}_{k}}=\int k(z,z')d\nu(z)d\nu'(z').
\]
Since $\mu_{k}$ is injective, this is a valid inner product and induces
a norm on $\mathcal{M}(\mathcal{Z})$, for which $\left\Vert \nu\right\Vert _{k}=\left\langle \left\langle \nu,\nu\right\rangle \right\rangle _{k}^{1/2}=0$
if and only if $\nu=0$. This fact has been used extensively in the
literature to formulate: (a) a nonparametric two-sample test based
on estimation of \emph{maximum mean discrepancy} $\left\Vert P-Q\right\Vert _{k}$,
for samples $\left\{ X_{i}\right\} _{i=1}^{n}\overset{i.i.d.}{\sim}P$,
$\left\{ Y_{i}\right\} _{i=1}^{m}\overset{i.i.d.}{\sim}Q$ \cite{Gretton2012}
and (b) a nonparametric independence test based on estimation of $\left\Vert P_{XY}-P_{X}P_{Y}\right\Vert _{k\otimes l}$,
for a joint sample $\left\{ \left(X_{i},Y_{i}\right)\right\} _{i=1}^{n}\overset{i.i.d.}{\sim}P_{XY}$
\cite{GreFukTeoSonetal08} (the latter is also called a Hilbert-Schmidt
independence criterion), with kernel $k\otimes l$ on the product
space defined as $k(x,x')l(y,y')$. When a bounded characteristic
kernel is used, the above tests are \emph{consistent against all alternatives},
and their alternative interpretation is as a generalization \cite{Sejdinovic2012,Sejdinovic2012a}
of energy distance \cite{Szekely2004,Baringhaus2004} and distance
covariance \cite{Szekely2007,SzeRiz09}.

In this article, we extend this approach to the three-variable case,
and formulate tests for both the Lancaster interaction and for the
total independence, using simple consistent estimators of $\left\Vert \Delta_{L}P\right\Vert _{k\otimes l\otimes m}$
and $\left\Vert \Delta_{tot}P\right\Vert _{k\otimes l\otimes m}$
respectively, which we describe in the next Section. Using the same
arguments as in the tests of \cite{Gretton2012,GreFukTeoSonetal08},
these tests are also consistent against all alternatives as long as
ISPD kernels are used.

\section{Interaction tests\label{sec:Interaction-tests}}

\textbf{Notational remarks:} Throughout the paper, $\circ$ denotes
an Hadamard (entrywise) product. Let $A$ be an $n\times n$ matrix,
and $K$ a symmetric $n\times n$ matrix. We will fix the following
notational conventions: $\one$ denotes an $n\times1$ column of ones;
$A_{+j}=\sum_{i=1}^{n}A_{ij}$ denotes the sum of all elements of
the $j$-th column of $A$; $A_{i+}=\sum_{j=1}^{n}A_{ij}$ denotes
the sum of all elements of the $i$-th row of $A$; $A_{++}=\sum_{i=1}^{n}\sum_{j=1}^{n}A_{ij}$
denotes the sum of all elements of $A$; $K_{+}=\one\onet K$, i.e.,
$\left[K_{+}\right]_{ij}=K_{+j}=K_{j+}$, and $\left[K_{+}^{\top}\right]_{ij}=K_{i+}=K_{+i}.$

\subsection{Two-Variable (Independence) Test}

We provide a short overview of the kernel independence test of \cite{GreFukTeoSonetal08},
which we write as the RKHS norm of the embedding of a signed measure.
While this material is not new (it appears in \cite[Section 7.4]{Gretton2012}),
it will help define how to proceed when a third variable is introduced,
and the signed measures become more involved. We begin by expanding
the squared RKHS norm $\left\Vert P_{XY}-P_{X}P_{Y}\right\Vert _{k\otimes l}^{2}$
as inner products, and applying the reproducing property, 
\begin{eqnarray}
\left\Vert P_{XY}-P_{X}P_{Y}\right\Vert _{k\otimes l}^{2} & = & \mathbb{E}_{XY}\mathbb{E}_{X'Y'}k(X,X')l(Y,Y')\;+\mathbb{E}_{X}\mathbb{E}_{X'}k(X,X')\mathbb{E}_{Y}\mathbb{E}_{Y'}l(Y,Y')\nonumber \\
 &  & \qquad-\,2\mathbb{E}_{X'Y'}\left[\mathbb{E}_{X}k(X,X')\mathbb{E}_{Y}l(Y,Y')\right],\label{eq: HSIC}
\end{eqnarray}
where $(X,Y)$ and $(X',Y')$ are independent copies of random variables
on $\mathcal{X}\times\mathcal{Y}$ with distribution $P_{XY}$.

Given a joint sample $\left\{ \left(X_{i},Y_{i}\right)\right\} _{i=1}^{n}\overset{i.i.d.}{\sim}P_{XY}$,
an empirical estimator of $\left\Vert P_{XY}-P_{X}P_{Y}\right\Vert _{k\otimes l}^{2}$
is obtained by substituting corresponding empirical means into \eqref{eq: HSIC},
which can be represented using Gram matrices $K$ and $L$ ($K_{ij}=k(X_{i},X_{j})$,
$L_{ij}=l(Y_{i},Y_{j})$),

\begin{align*}
\hat{\mathbb{E}}_{XY}\hat{\mathbb{E}}_{X'Y'}k(X,X')l(Y,Y') & =\frac{1}{n^{2}}\sum_{a=1}^{n}\sum_{b=1}^{n}K_{ab}L_{ab}=\frac{1}{n^{2}}\left(K\circ L\right)_{++},\\
\mathbb{\hat{E}}_{X}\hat{\mathbb{E}}_{X'}k(X,X')\hat{\mathbb{E}}_{Y}\hat{\mathbb{E}}_{Y'}l(Y,Y') & =\frac{1}{n^{4}}\sum_{a=1}^{n}\sum_{b=1}^{n}\sum_{c=1}^{n}\sum_{d=1}^{n}K_{ab}L_{cd}=\frac{1}{n^{4}}K_{++}L_{++},\\
\hat{\mathbb{E}}_{X'Y'}\left[\hat{\mathbb{E}}_{X}k(X,X')\hat{\mathbb{E}}_{Y}l(Y,Y')\right] & =\frac{1}{n^{3}}\sum_{a=1}^{n}\sum_{b=1}^{n}\sum_{c=1}^{n}K_{ac}L_{bc}=\frac{1}{n^{3}}\left(KL\right)_{++}.
\end{align*}

\begin{table}
\caption{\label{tab: 2vartab}{\footnotesize $V$-statistic estimates of $\left\langle \left\langle \nu,\nu'\right\rangle \right\rangle _{k\otimes l}$
in the two-variable case}}

\centering{}{\small }%
\begin{tabular}{|c|c|c|}
\hline 
{\small $\nu\backslash\nu'$} & {\small $P_{XY}$} & {\small $P_{X}P_{Y}$}\tabularnewline
\hline 
\hline 
{\small $P_{XY}$} & {\small $\frac{1}{n^{2}}\left(K\circ L\right)_{++}$} & {\small $\frac{1}{n^{3}}\left(KL\right)_{++}$}\tabularnewline
\hline 
{\small $P_{X}P_{Y}$} &  & {\small $\frac{1}{n^{4}}K_{++}L_{++}$}\tabularnewline
\hline 
\end{tabular}
\end{table}
Since these are V-statistics, there is a bias of $O_{P}(n^{-1})$;
U-statistics may be used if an unbiased estimate is needed. Each of
the terms above corresponds to an estimate of an inner product $\left\langle \left\langle \nu,\nu'\right\rangle \right\rangle _{k\otimes l}$
for probability measures $\nu$ and $\nu'$ taking values in $\left\{ P_{XY},P_{X}P_{Y}\right\} $,
as summarized in Table \ref{tab: 2vartab}. Even though the second
and third terms involve triple and quadruple sums, each of the empirical
means can be computed using sums of all terms of certain matrices,
where the dominant computational cost is in computing the matrix product
$KL$. In fact, the overall estimator can be computed in an even simpler
form (see Proposition \ref{prop: 2varcase-2} in Appendix \ref{sec: matrix_algebra}),
as $\left\Vert \hat{P}_{XY}-\hat{P}_{X}\hat{P}_{Y}\right\Vert _{k\otimes l}^{2}=\frac{1}{n^{2}}\left(K\circ HLH\right)_{++},$
where $H=I-\frac{1}{n}\one\onet$ is the centering matrix. Note that
by the idempotence of $H$, we also have that $\left(K\circ HLH\right)_{++}=\left(HKH\circ HLH\right)_{++}$.
In the rest of the paper, for any Gram matrix $K$, we will denote
its corresponding centered matrix $HKH$ by $\tilde{K}$. When three
variables are present, a two-variable test already allows us to determine
whether for instance $(X,Y)\ci Z$, i.e., whether $P_{XYZ}=P_{XY}P_{Z}$.
It is sufficient to treat $(X,Y)$ as a single variable on the product
space $\mathcal{X}\times\mathcal{Y}$, with the product kernel $k\otimes l$.
Then, the Gram matrix associated to $(X,Y)$ is simply $K\circ L$,
and the corresponding $V$-statistic is $\frac{1}{n^{2}}\left(K\circ L\circ\tilde{M}\right)_{++}$.%
\footnote{In general, however, this approach would require some care since,
e.g., $X$ and $Y$ could be measured on very different scales, and
the choice of kernels $k$ and $l$ needs to take this into account.%
} What is not obvious, however, is if a V-statistic for the Lancaster
interaction (which can be thought of as a surrogate for the composite
hypothesis of various factorizations) can be obtained in a similar
form. We will address this question in the next section.

\subsection{Three-Variable tests}

As in the two-variable case, it suffices to derive V-statistics for
inner products $\left\langle \left\langle \nu,\nu'\right\rangle \right\rangle _{k\otimes l\otimes m}$,
where $\nu$ and $\nu'$ take values in all possible combinations
of the joint and the products of the marginals, i.e., $P_{XYZ}$,
$P_{XY}P_{Z}$, etc. Again, it is easy to see that these can be expressed
as certain expectations of kernel functions, and thereby can be calculated
by an appropriate manipulation of the three Gram matrices. We summarize
the resulting expressions in Table \ref{tab: 3vartab} - their derivation
is a tedious but straightforward linear algebra exercise. For compactness,
the appropriate normalizing terms are moved inside the measures considered.

\begin{table}
\caption{\label{tab: 3vartab}{\footnotesize $V$-statistic estimates of $\left\langle \left\langle \nu,\nu'\right\rangle \right\rangle _{k\otimes l\otimes m}$
in the three-variable case}}

\centering{}\textbf{\scriptsize }%
\begin{tabular}{|c|c|c|c|c|c|}
\hline 
{\scriptsize $\nu\backslash\nu'$} & \textbf{\scriptsize $nP_{XYZ}$} & \textbf{\scriptsize $n^{2}P_{XY}P_{Z}$} & \textbf{\scriptsize $n^{2}P_{XZ}P_{Y}$} & \textbf{\scriptsize $n^{2}P_{YZ}P_{X}$} & \textbf{\scriptsize $n^{3}P_{X}P_{Y}P_{Z}$}\tabularnewline
\hline 
\hline 
\textbf{\scriptsize $nP_{XYZ}$} & \textbf{\scriptsize $\left(K\circ L\circ M\right)_{++}$} & \textbf{\scriptsize $\left(\left(K\circ L\right)M\right)_{++}$} & \textbf{\scriptsize $\left(\left(K\circ M\right)L\right)_{++}$} & \textbf{\scriptsize $\left(\left(M\circ L\right)K\right)_{++}$} & \textbf{\scriptsize $tr(K_{+}\circ L_{+}\circ M_{+})$}\tabularnewline
\hline 
\textbf{\scriptsize $n^{2}P_{XY}P_{Z}$} &  & \textbf{\scriptsize $\left(K\circ L\right)_{++}M_{++}$} & \textbf{\scriptsize $\left(MKL\right)_{++}$} & \textbf{\scriptsize $\left(KLM\right)_{++}$} & \textbf{\scriptsize $(KL)_{++}M_{++}$}\tabularnewline
\hline 
\textbf{\scriptsize $n^{2}P_{XZ}P_{Y}$} &  &  & \textbf{\scriptsize $\left(K\circ M\right)_{++}L_{++}$} & \textbf{\scriptsize $\left(KML\right)_{++}$} & \textbf{\scriptsize $(KM)_{++}L_{++}$}\tabularnewline
\hline 
\textbf{\scriptsize $n^{2}P_{YZ}P_{X}$} &  &  &  & \textbf{\scriptsize $\left(L\circ M\right)_{++}K_{++}$} & \textbf{\scriptsize $(LM)_{++}K_{++}$}\tabularnewline
\hline 
\textbf{\scriptsize $n^{3}P_{X}P_{Y}P_{Z}$} &  &  &  &  & \textbf{\scriptsize $K_{++}L_{++}M_{++}$}\tabularnewline
\hline 
\end{tabular}
\end{table}

Based on the individual RKHS inner product estimators, we can now
easily derive estimators for various signed measures arising as linear
combinations of $P_{XYZ},P_{XY}P_{Z},$ and so on. The first such
measure is an {}``incomplete'' Lancaster interaction measure $\Delta_{(Z)}P=P_{XYZ}+P_{X}P_{Y}P_{Z}-P_{YZ}P_{X}-P_{XZ}P_{Y}$,
which vanishes if $(Y,Z)\ci X$ or $(X,Z)\ci Y$, but not necessarily
if $(X,Y)\ci Z$. We obtain the following result for the empirical
measure $\hat{P}$.
\begin{prop}
[Incomplete Lancaster interaction] \label{prop: Incomplete-Lancaster-interaction}
$\left\Vert \Delta_{(Z)}\hat{P}\right\Vert _{k\otimes l\otimes m}^{2}=\frac{1}{n^{2}}\left(\tilde{K}\circ\tilde{L}\circ M\right)_{++}.$
\end{prop}

Analogous expressions hold for $\Delta_{(X)}\hat{P}$ and $\Delta_{(Y)}\hat{P}$.
Unlike in the two-variable case where either matrix or both can be
centered, centering of each matrix in the three-variable case has
a different meaning. In particular, one requires centering of all
three kernel matrices to perform a {}``complete'' Lancaster interaction
test, as given by the following Proposition. 
\begin{prop}
[Lancaster interaction] \label{prop: Lancaster-interaction}$\left\Vert \Delta_{L}\hat{P}\right\Vert _{k\otimes l\otimes m}^{2}=\frac{1}{n^{2}}\left(\tilde{K}\circ\tilde{L}\circ\tilde{M}\right)_{++}.$
\end{prop}
The proofs of these Propositions are given in Appendix \ref{sec: Proofs}.
We summarize various hypotheses and the associated V-statistics in
the Appendix \ref{sec: centering}. As we will demonstrate in the
experiments in Section \ref{sec:Experiments}, while particularly
useful for testing the factorization hypothesis, i.e., for $(X,Y)\ci Z\,\vee\,(X,Z)\ci Y\,\vee\,(Y,Z)\ci X$,
the statistic $\left\Vert \Delta_{L}\hat{P}\right\Vert _{k\otimes l\otimes m}^{2}$
can also be used for powerful tests of either the individual hypotheses
$(Y,Z)\ci X$, $(X,Z)\ci Y$, or $(X,Y)\ci Z$, or for total independence
testing, i.e., $P_{XYZ}=P_{X}P_{Y}P_{Z}$, as it vanishes in all of
these cases. The null distribution under each of these hypotheses
can be estimated using a standard permutation-based approach described
in Appendix \ref{sec: Permutation test}.

Another way to obtain the Lancaster interaction statistic is as the
RKHS norm of the joint {}``central moment'' $\Sigma_{XYZ}=\mathbb{\mathbb{\mathbb{E}}}_{XYZ}[\left(k_{X}-\mu_{X}\right)\otimes\left(l_{Y}-\mu_{Y}\right)\otimes\left(m_{Z}-\mu_{Z}\right)]$
of RKHS-valued random variables $k_{X}$, $l_{Y}$ and $m_{Z}$ (understood
as an element of the tensor RKHS $\mathcal{H}_{k}\otimes\mathcal{H}_{l}\otimes\mathcal{H}_{m}$).
This is related to a classical characterization of the Lancaster interaction
\cite[Ch. XII]{Lancaster1969}: there is no Lancaster interaction
between $X$, $Y$ and $Z$ if and only if $\textrm{cov}\left[f(X),g(Y),h(Z)\right]=0$
for all $L_{2}$ functions $f$, $g$ and $h$. There is an analogous
result in our case (proof is given in Appendix \ref{sec: Proofs}),
which states
\begin{prop}
\label{prop: cov_RKHS_functions}$\left\Vert \Delta_{L}P\right\Vert _{k\otimes l\otimes m}=0$
if and only if $\textrm{cov}\left[f(X),g(Y),h(Z)\right]=0$ for all
$f\in\mathcal{H}_{k}$, $g\in\mathcal{H}_{l}$, $h\in\mathcal{H}_{m}$.
\end{prop}
And finally, we give an estimator of the RKHS norm of the total independence
measure $\Delta_{tot}P$.
\begin{prop}
[Total independence]Let $\Delta_{tot}\hat{P}=\hat{P}_{XYZ}-\hat{P}_{X}\hat{P}_{Y}\hat{P}_{Z}$.
Then:
\begin{eqnarray*}
\left\Vert \Delta_{tot}\hat{P}\right\Vert _{k\otimes l\otimes m}^{2} & = & \frac{1}{n^{2}}\left(K\circ L\circ M\right)_{++}-\frac{2}{n^{4}}tr(K_{+}\circ L_{+}\circ M_{+})+\frac{1}{n^{6}}K_{++}L_{++}M_{++}.
\end{eqnarray*}

\end{prop}
The proof follows simply from reading off the corresponding inner-product
V-statistics from the Table \ref{tab: 3vartab}. While the test statistic
for total independence has a somewhat more complicated form than that
of Lancaster interaction, it can also be computed in quadratic time.

\subsection{Interaction for $D>3$ \label{sec:The-case-D-bigger-3}}

Streitberg's correction of the interaction measure for $D>3$ has
the form
\begin{equation}
\Delta_{S}P=\sum_{\pi}(-1)^{\left|\pi\right|-1}\left(\left|\pi\right|-1\right)!J_{\pi}P,
\end{equation}
where the sum is taken over all partitions of the set $\left\{ 1,2,\ldots,n\right\} $,
$\left|\pi\right|$ denotes the size of the partition (number of blocks),
and $J_{\pi}:P\mapsto P_{\pi}$ is the \emph{partition operator} on
probability measures, which for a fixed partition $\pi=\pi_{1}|\pi_{2}|\ldots|\pi_{r}$
maps the probability measure $P$ to the product measure $P_{\pi}=\prod_{j=1}^{r}P_{\pi_{j}}$,
where $P_{\pi_{j}}$ is the marginal distribution of the subvector
$\left(X_{i}\,:\, i\in\pi_{j}\right).$ The coefficients correspond
to the M\"{o}bius inversion on the partition lattice \cite{Speed1983}.
While the Lancaster interaction has an interpretation in terms of
joint central moments, Streitberg's correction corresponds to joint
cumulants \cite[Section 4]{Streitberg1990}. Therefore, a central
moment expression like $\mathbb{\mathbb{\mathbb{E}}}_{X_{1}\ldots X_{n}}[\left(k_{X_{1}}^{(1)}-\mu_{X_{1}}\right)\otimes\cdots\otimes\left(k_{X_{n}}^{(n)}-\mu_{X_{n}}\right)]$
does not capture the correct notion of the interaction measure. Thus,
while one can in principle construct RKHS embeddings of higher-order
interaction measures, and compute RKHS norms using a calculus of $V$-statistics
and Gram-matrices analogous to that of Table \ref{tab: 3vartab},
it does not seem possible to avoid summing over all partitions when
computing the corresponding statistics, yielding a computationally
prohibitive approach in general. This can be viewed by analogy with
the scalar case, where it is well known that the second and third
cumulants coincide with the second and third central moments, whereas
the higher order cumulants are neither moments nor central moments,
but some other polynomials of the moments.

\subsection{Total independence for $D>3$}

In general, the test statistic for total independence in the $D$-variable
case is 
\begin{eqnarray*}
\left\Vert \hat{P}_{X_{1:D}}-\prod_{i=1}^{D}\hat{P}_{X_{i}}\right\Vert _{\bigotimes_{i=1}^{D}k^{(i)}}^{2} & = & \frac{1}{n^{2}}\sum_{a=1}^{n}\sum_{b=1}^{n}\prod_{i=1}^{D}K_{ab}^{(i)}-\frac{2}{n^{D+1}}\sum_{a=1}^{n}\prod_{i=1}^{D}\sum_{b=1}^{n}K_{ab}^{(i)}\\
 & + & \frac{1}{n^{2D}}\prod_{i=1}^{D}\sum_{a=1}^{n}\sum_{b=1}^{n}K_{ab}^{(i)}.
\end{eqnarray*}
A similar statistic for total independence is discussed by \cite{Kankainen1995}
where testing of total independence based on empirical characteristic
functions is considered. Our test has a direct interpretation in terms
of characteristic functions as well, which is straightforward to see
in the case of translation invariant kernels on Euclidean spaces,
using their Bochner representation, similarly as in \cite[Corollary 4]{SriGreFukLanetal10}.

\section{Experiments \label{sec:Experiments}}

We investigate the performance of various permutation based tests
that use the Lancaster statistic $\left\Vert \Delta_{L}\hat{P}\right\Vert _{k\otimes l\otimes m}^{2}$
and the total independence statistic $\left\Vert \Delta_{tot}\hat{P}\right\Vert _{k\otimes l\otimes m}^{2}$
on two synthetic datasets where $X$, $Y$ and $Z$ are random vectors
of increasing dimensionality:

\textbf{Dataset A: Pairwise independent, mutually dependent data.}
Our first dataset is a triplet of random vectors $(X,Y,Z)$ on $\mathbb{R}^{p}\times\mathbb{R}^{p}\times\mathbb{R}^{p}$,
with $X,Y\iid\mathcal{N}(0,I_{p})$, $W\sim Exp(\frac{1}{\sqrt{2}})$,
$Z_{1}=sign(X_{1}Y_{1})W$, and $Z_{\ensuremath{2:p}}\sim\mathcal{N}(0,I_{p-1})$,
i.e., the product of $X_{1}Y_{1}$ determines the sign of $Z_{1}$,
while the remaining $p-1$ dimensions are independent (and serve as
noise in this example).%
\footnote{Note that there is no reason for $X$, $Y$ and $Z$ to have the same
dimensionality $p$ - this is done for simplicity of exposition.%
} In this case, $(X,Y,Z)$ is clearly a pairwise independent but mutually
dependent triplet. The mutual dependence becomes increasingly difficult
to detect as the dimensionality $p$ increases.

\begin{figure}
\centering
\psfragfig{}
\psfragfig{}

\caption{\label{fig: MarginalPlot} Two-variable kernel independence tests
and the test for $(X,Y)\ci Z$ using the Lancaster statistic}
\end{figure}

\begin{figure}
\centering
\psfragfig{}
\psfragfig{}

\caption{\label{fig: TotalPlot} Total independence: $\Delta_{tot}\hat{P}$
vs. $\Delta_{L}\hat{P}$ .}
\end{figure}

\begin{figure}
\centering
\psfragfig{}
\psfragfig{}

\caption{\label{fig: VstructPlot}Factorization hypothesis: Lancaster statistic
vs. a two-variable based test; Test for $X\ci Y|Z$ from \cite{Zhang2011}}
\end{figure}

\textbf{Dataset B: Joint dependence can be easier to detect.} In this
example, we consider a triplet of random vectors $(X,Y,Z)$ on $\mathbb{R}^{p}\times\mathbb{R}^{p}\times\mathbb{R}^{p}$,
with $X,Y\iid\mathcal{N}(0,I_{p})$, $Z_{\ensuremath{2:p}}\sim\mathcal{N}(0,I_{p-1})$,
and{\small 
\[
Z_{1}=\begin{cases}
X_{1}^{2}+\epsilon, & w.p.\;1/3,\\
Y_{1}^{2}+\epsilon, & w.p.\;1/3,\\
X_{1}Y_{1}+\epsilon, & w.p.\;1/3,
\end{cases}
\]
}where $\epsilon\sim\mathcal{N}(0,0.1^{2})$. Thus, dependence of
$Z$ on pair $(X,Y)$ is stronger than on $X$ and $Y$ individually.

In all cases, we use permutation tests as described in Appendix \ref{sec: Permutation test}.
The test level is set to $\alpha=0.05$, and we use gaussian kernels
with bandwidth set to the interpoint median distance. In Figure \ref{fig: MarginalPlot},
we plot the null hypothesis acceptance rates of the standard kernel
two-variable tests for $X\ci Y$ (which is true for both datasets
A and B, and accepted at the correct rate across all dimensions) and
for $X\ci Z$ (which is true only for dataset A), as well as of the
standard kernel two-variable test for $(X,Y)\ci Z$, and the test
for $(X,Y)\ci Z$ using the Lancaster statistic. As expected, in dataset
B, we see that dependence of $Z$ on pair $(X,Y)$ is somewhat easier
to detect than on $X$ individually with two-variable tests. In both
datasets, however, the Lancaster interaction appears significantly
more sensitive in detecting this dependence as dimensionality $p$
increases. Figure \ref{fig: TotalPlot} plots the Type II error of
total independence tests with statistics $\left\Vert \Delta_{L}\hat{P}\right\Vert _{k\otimes l\otimes m}^{2}$
and $\left\Vert \Delta_{tot}\hat{P}\right\Vert _{k\otimes l\otimes m}^{2}$.
The Lancaster statistic outperforms the total independence statistic
everywhere apart from the Dataset B when the number of dimensions
is small (between 1 and 5). Figure \ref{fig: TotalPlot} plots the
Type II error of the factorization test, i.e., test for $(X,Y)\ci Z\,\vee\,(X,Z)\ci Y\,\vee\,(Y,Z)\ci X$
with Lancaster statistic with Holm-Bonferroni correction as described
in Appendix \ref{sec: Permutation test}, as well as the two-variable
based test (which performs three standard two-variable tests and applies
the Holm-Bonferroni correction). We also plot the Type II error for
the conditional independence test for $X\ci Y|Z$ from \cite{Zhang2011}.
Under assumption that $X\ci Y$ (correct on both datasets), negation
of each of these three hypotheses is equivalent to the presence of
V-structure $X\to Z\leftarrow Y$, so the rejection of the null can
be viewed as a V-structure detection procedure. As dimensionality
increases, the Lancaster statistic appears significantly more sensitive
to the interactions present than the competing approaches, which is
particularly pronounced in Dataset A.

\section{Conclusions}

We have constructed permutation-based nonparametric tests for three-variable
interactions, including the Lancaster interaction and total independence.
The tests can be used in datasets where only higher-order interactions
persist, i.e., variables are pairwise independent; as well as in cases
where joint dependence may be easier to detect than pairwise dependence,
for instance when the effect of two variables on a third is not additive.
The flexibility of the framework of RKHS embeddings of signed measures
allows us to consider variables that are themselves multidimensional.
While the total independence case readily generalizes to more than
three dimensions, the combinatorial nature of joint cumulants implies
that detecting interactions of higher order requires significantly
more costly computation, and is an interesting topic for future work.

\appendix

\section{\label{sec: Proofs}Proofs}

\subsection{Proof of Proposition \ref{prop: Incomplete-Lancaster-interaction}}

Some basic matrix algebra used in this proof is reviewed in Appendix
\ref{sec: matrix_algebra}. The proof of the following simple Lemma
directly follows from the results therein.
\begin{lem}
The following equalities hold:
\begin{enumerate}
\item \textup{$\left(K_{+}\circ L_{+}\circ M\right)_{++}=\left(K_{+}^{\top}\circ L_{+}^{\top}\circ M\right)_{++}=tr(K_{+}\circ L_{+}\circ M_{+})=\sum_{a=1}^{n}K_{a+}L_{a+}M_{a+}$}
\item \textup{$\left(K_{+}\circ L\circ M_{+}^{\top}\right)_{++}=\left(KLM\right)_{++}$}
\end{enumerate}
\end{lem}
Now, we will take a kernel matrix $M$ and consider its Hadamard product
with $\tilde{K}\circ\tilde{L}$:

\begin{eqnarray*}
\tilde{K}\circ\tilde{L}\circ M & = & K\circ L\circ M-\frac{1}{n}\left[\underset{A}{\underbrace{K\circ L_{+}\circ M}}+\underset{A^{\top}}{\underbrace{K\circ L_{+}^{\top}\circ M}}+\underset{B}{\underbrace{K_{+}\circ L\circ M}}+\underset{B^{\top}}{\underbrace{K_{+}^{\top}\circ L\circ M}}\right]\\
 & + & \frac{1}{n^{2}}\left(K_{++}L\circ M+L_{++}K\circ M\right)\\
 & + & \frac{1}{n^{2}}\left[\underset{C}{\underbrace{K_{+}\circ L_{+}\circ M}}+\underset{C^{\top}}{\underbrace{K_{+}^{\top}\circ L_{+}^{\top}\circ M}}+\underset{D}{\underbrace{K_{+}\circ L_{+}^{\top}\circ M}}+\underset{D^{\top}}{\underbrace{K_{+}^{\top}\circ L_{+}\circ M}}\right]\\
 & - & \frac{1}{n^{3}}K_{++}\left[L_{+}\circ M+L_{+}^{\top}\circ M\right]-\frac{1}{n^{3}}L_{++}\left[K_{+}\circ M+K_{+}^{\top}\circ M\right]\\
 & + & \frac{1}{n^{4}}K_{++}L_{++}M.
\end{eqnarray*}
and thus:
\begin{eqnarray*}
\left(\tilde{K}\circ\tilde{L}\circ M\right)_{++} & = & \left(K\circ L\circ M\right)_{++}-\frac{2}{n}\left(\left(K\circ M\right)L+\left(L\circ M\right)K\right)_{++}\\
 & + & \frac{1}{n^{2}}\left[K_{++}(L\circ M)_{++}+L_{++}(K\circ M)_{++}\right]\\
 & + & \frac{2}{n^{2}}\left[tr(K_{+}\circ L_{+}\circ M_{+})+\left(LMK\right)_{++}\right]\\
 & - & \frac{2}{n^{3}}\left[K_{++}\left(LM\right)_{++}+L_{++}(KM)_{++}\right]\\
 & + & \frac{1}{n^{4}}K_{++}L_{++}M_{++}.
\end{eqnarray*}
where we used that $A_{++}=\left(\left(K\circ M\right)\circ L_{+}\right)_{++}=\left(\left(K\circ M\right)L\right)_{++},$
and similarly $B_{++}=\left(\left(L\circ M\right)K\right)_{++}.$
Also, $C_{++}=tr(K_{+}\circ L_{+}\circ M_{+})$ and $D_{++}=(LMK)_{++}$.

By comparing to the table of V-statistics, we obtain that:
\begin{eqnarray*}
\frac{1}{n^{2}}\left(\tilde{K}\circ\tilde{L}\circ M\right)_{++} & = & \left\Vert \Delta_{(Z)}\hat{P}\right\Vert _{k\otimes l\otimes m}^{2}
\end{eqnarray*}
where $\Delta_{(Z)}\hat{P}=\hat{P}_{XYZ}+\hat{P}_{X}\hat{P}_{Y}\hat{P}_{Z}-\hat{P}_{YZ}\hat{P}_{X}-\hat{P}_{XZ}\hat{P}_{Y}$,
which completes the proof of Proposition \ref{prop: Incomplete-Lancaster-interaction}.
Proposition \ref{prop: Lancaster-interaction} can be proved in an
analogous way by including the additional terms corresponding to centering
of $M$, i.e., $\left(\tilde{K}\circ\tilde{L}\circ M_{+}\right)_{++}$
and $\left(\tilde{K}\circ\tilde{L}\circ M_{++}\right)_{++}$.  In
the next Section, however, we give an alternative proof which gives
more insight into the role that the centering of each Gram matrix
plays.

\subsection{\label{sub:Proof-of-Proposition A2}Proof of Proposition \ref{prop: Lancaster-interaction}}

It will be useful to introduce into notation the kernel centered at
a probability measure $\nu$, given by: 
\begin{equation}
\tilde{k}_{\nu}(z,z'):=k(z,z')+\int\int k(w,w')d\nu(w)d\nu(w)-\int\left[k(z,w)+k(z',w)\right]d\nu(w),\label{eq: centred_kernels}
\end{equation}
Note that $\int\tilde{k}_{\nu}(z,z')d\nu(z)d\nu(z')=0$, i.e., $\mu_{\tilde{k}_{\nu}}(\nu)\equiv0$.

By expanding the population expression of the kernel norm of the joint
under the kernels centered at the marginals, we obtain: 
\begin{eqnarray*}
 &  & \left\Vert P_{XYZ}\right\Vert _{\tilde{k}_{P_{X}}\otimes\tilde{l}_{P_{Y}}\otimes\tilde{m}_{P_{Z}}}^{2}\\
 &  & \quad=\int\int\left[\tilde{k}_{P_{X}}(x,x')\tilde{l}_{P_{Y}}(y,y')\tilde{m}_{P_{Z}}(z,z')\right]\\
 &  & \qquad\qquad dP_{XYZ}(x,y,z)dP_{XYZ}(x',y',z'),
\end{eqnarray*}
Substituting the definition of the centered kernel in \eqref{eq: centred_kernels},
it is readily obtained that 
\begin{eqnarray*}
\left\Vert P_{XYZ}\right\Vert _{\tilde{k}_{P_{X}}\otimes\tilde{l}_{P_{Y}}\otimes\tilde{m}_{P_{Z}}}^{2} & = & \left\Vert \Delta_{L}P\right\Vert _{k\otimes l\otimes m}^{2}.
\end{eqnarray*}
Now, $\left\Vert P_{XYZ}\right\Vert _{\tilde{k}_{P_{X}}\otimes\tilde{l}_{P_{Y}}\otimes\tilde{m}_{P_{Z}}}^{2}$
is the first term in the expansion of $\left\Vert \Delta_{L}P\right\Vert _{\tilde{k}_{P_{X}}\otimes\tilde{l}_{P_{Y}}\otimes\tilde{m}_{P_{Z}}}^{2}$.
Let us show that all the other terms are equal to zero. Indeed, all
the other terms are of the form
\begin{eqnarray*}
\left\langle \left\langle P_{W}Q,Q'\right\rangle \right\rangle _{\tilde{k}_{P_{X}}\otimes\tilde{l}_{P_{Y}}\otimes\tilde{m}_{P_{Z}}},
\end{eqnarray*}
where $W=X$, $Y$, or $Z$ (individual variable). Without loss of
generality, let $W=X$. Then,

\begin{eqnarray*}
 &  & \left\langle \left\langle P_{X}Q,Q'\right\rangle \right\rangle _{\tilde{k}_{P_{X}}\otimes\tilde{l}_{P_{Y}}\otimes\tilde{m}_{P_{Z}}}\\
 &  & \quad=\int\int\int\left[\tilde{k}_{P_{X}}(x,x')\tilde{l}_{P_{Y}}(y,y')\tilde{m}_{P_{Z}}(z,z')\right]\\
 &  & \qquad\qquad\qquad\qquad dP_{X}(x)dQ(y,z)dQ'(x',y',z')\\
 &  & \quad=\int\int\underset{=\left[\mu_{\tilde{k}_{P_{X}}}(P_{X})\right](x')=0}{\underbrace{\int\tilde{k}_{P_{X}}(x,x')dP_{X}(x)}}\tilde{l}_{P_{Y}}(y,y')\tilde{m}_{P_{Z}}(z,z')\\
 &  & \qquad\qquad\qquad\qquad\qquad dQ(y,z)dQ'(x',y',z')\\
 &  & \quad=0.
\end{eqnarray*}
Therefore, 
\begin{eqnarray*}
\left\Vert \Delta_{L}P\right\Vert _{\tilde{k}_{P_{X}}\otimes\tilde{l}_{P_{Y}}\otimes\tilde{m}_{P_{Z}}}^{2} & = & \left\Vert P_{XYZ}\right\Vert _{\tilde{k}_{P_{X}}\otimes\tilde{l}_{P_{Y}}\otimes\tilde{m}_{P_{Z}}}^{2}\\
 & = & \left\Vert \Delta_{L}P\right\Vert _{k\otimes l\otimes m}^{2}.
\end{eqnarray*}
The above is true for any joint distribution $P_{XYZ}$, and in particular
for the empirical joint, whereby:
\begin{eqnarray*}
\left\Vert \Delta_{L}\hat{P}\right\Vert _{k\otimes l\otimes m}^{2} & = & \left\Vert \hat{P}_{XYZ}\right\Vert _{\tilde{k}_{\hat{P}_{X}}\otimes\tilde{l}_{\hat{P}_{Y}}\otimes\tilde{m}_{\hat{P}_{Z}}}^{2}\\
 & = & \frac{1}{n^{2}}\left(\tilde{K}\circ\tilde{L}\circ\tilde{M}\right)_{++}.
\end{eqnarray*}

\subsection{Proof of Proposition \ref{prop: cov_RKHS_functions}}

Consider the element of $\mathcal{H}_{k}\otimes\mathcal{H}_{l}\otimes\mathcal{H}_{m}$
given by $\mathbb{E}_{XYZ}k(\cdot,X)\otimes l(\cdot,Y)\otimes m(\cdot,Z)$.
This can be identified with a Hilbert-Schmidt uncentered covariance
operator $C_{(XY)Z}:\mathcal{H}_{k}\otimes\mathcal{H}_{l}\rightarrow\mathcal{H}_{m}$,
such that $\forall f\in\mathcal{H}_{k},g\in\mathcal{H}_{l},h\in\mathcal{H}_{m}$:
\begin{eqnarray*}
\left\langle C_{(XY)Z}\left[f\otimes g\right],h\right\rangle _{\mathcal{H}_{m}} & = & \mathbb{E}_{XYZ}f(X)g(Y)h(Z).
\end{eqnarray*}
By replacing $k$, $l$, $m$ with kernels centered at the marginals,
we obtain a centered covariance operator $\Sigma_{(XY)Z}$, for which 

\begin{eqnarray*}
\left\langle \Sigma_{(XY)Z}\left[f\otimes g\right],h\right\rangle _{\mathcal{H}_{m}} & = & \mathbb{E}_{XYZ}\tilde{f}(X)\tilde{g}(Y)\tilde{h}(Z)\\
 & = & \textrm{cov}\left[f(X),g(Y),h(Z)\right],
\end{eqnarray*}
where we wrote $\tilde{f}(X)=f(X)-\mathbb{E}f(X)$, and similarly
for $\tilde{g}$ and $\tilde{h}$. Using the usual isometries between
Hilbert-Schmidt spaces and the tensor product spaces:{\small 
\begin{eqnarray*}
 &  & \left\Vert \Sigma_{(XY)Z}\right\Vert _{HS}^{2}\\
 & = & \left\Vert \mathbb{E}_{XYZ}\tilde{k}_{P_{X}}(\cdot,X)\otimes\tilde{l}_{P_{Y}}(\cdot,Y)\otimes\tilde{m}_{P_{Z}}(\cdot,Z)\right\Vert _{\mathcal{H}_{k}\otimes\mathcal{H}_{l}\otimes\mathcal{H}_{m}}^{2}\\
 &  & \qquad\qquad=\left\Vert P_{XYZ}\right\Vert _{\tilde{k}_{P_{X}}\otimes\tilde{l}_{P_{Y}}\otimes\tilde{m}_{P_{Z}}}^{2}\\
 &  & \qquad\qquad\qquad=\left\Vert \Delta_{L}P\right\Vert _{k\otimes l\otimes m}^{2}.
\end{eqnarray*}
}Now, consider the supremum of the three-way covariance taken over
the unit balls of respective RKHSs:

\begin{eqnarray*}
\sup_{f,g,h}\textrm{cov}\left[f(X),g(Y),h(Z)\right] & = & \sup_{f,g,h}\left\langle \Sigma_{(XY)Z}\left[f\otimes g\right],h\right\rangle _{\mathcal{H}_{m}}\\
 & = & \sup_{f,g}\left\Vert \Sigma_{(XY)Z}\left[f\otimes g\right]\right\Vert _{\mathcal{H}_{m}}\\
 & \leq & \sup_{F\in\mathcal{H}_{k}\otimes\mathcal{H}_{l}}\left\Vert \Sigma_{(XY)Z}F\right\Vert _{\mathcal{H}_{m}}\\
 & = & \left\Vert \Sigma_{(XY)Z}\right\Vert _{op}\leq\left\Vert \Sigma_{(XY)Z}\right\Vert _{HS}.
\end{eqnarray*}
and thus, $\left\Vert \Delta_{L}P\right\Vert _{k\otimes l\otimes m}=0$
implies $\sup_{f,g,h}\textrm{cov}\left[f(X),g(Y),h(Z)\right]=0$.
Conversely, if $\textrm{cov}\left[f(X),g(Y),h(Z)\right]=0$ $\forall f,g,h$,
then $\Sigma_{(XY)Z}\left[f\otimes g\right]\equiv0$ $\forall f,g$,
so the linear operator $\Sigma_{(XY)Z}$ vanishes.

\section{\label{sec: centering}The effect of centering}

\begin{table}
\caption{{\footnotesize \label{tab: Vstat-summary}V-statistics for various
hypotheses }}

\centering{}{\scriptsize }%
\begin{tabular}{|c|c|c|c|}
\hline 
{\scriptsize hypothesis} & {\scriptsize V-statistic} & {\scriptsize hypothesis} & {\scriptsize V-statistic}\tabularnewline
\hline 
\hline 
{\scriptsize $(X,Y)\ci Z$} & {\scriptsize $\frac{1}{n^{2}}\left(K\circ L\circ\tilde{M}\right)_{++}$} & {\scriptsize $\Delta_{(X)}P=0$} & {\scriptsize $\frac{1}{n^{2}}\left(K\circ\tilde{L}\circ\tilde{M}\right)_{++}$}\tabularnewline
\hline 
{\scriptsize $(X,Z)\ci Y$} & {\scriptsize $\frac{1}{n^{2}}\left(K\circ\tilde{L}\circ M\right)_{++}$} & {\scriptsize $\Delta_{(Y)}P=0$} & {\scriptsize $\frac{1}{n^{2}}\left(\tilde{K}\circ L\circ\tilde{M}\right)_{++}$}\tabularnewline
\hline 
{\scriptsize $(Y,Z)\ci X$} & {\scriptsize $\frac{1}{n^{2}}\left(\tilde{K}\circ L\circ M\right)_{++}$} & {\scriptsize $\Delta_{(Z)}P=0$} & {\scriptsize $\frac{1}{n^{2}}\left(\tilde{K}\circ\tilde{L}\circ M\right)_{++}$}\tabularnewline
\hline 
 &  & {\scriptsize $\Delta_{L}P=0$} & {\scriptsize $\frac{1}{n^{2}}\left(\tilde{K}\circ\tilde{L}\circ\tilde{M}\right)_{++}$}\tabularnewline
\hline 
\end{tabular}
\end{table}

In a two-variable test, either or both of the kernel matrices can
be centered when computing the test statistic since $\left(K\circ\tilde{L}\right)_{++}=\left(\tilde{K}\circ L\right)_{++}=\left(\tilde{K}\circ\tilde{L}\right)_{++}$.
To see this, simply note that by the idempotence of $H$,

\begin{eqnarray}
\left(K\circ\tilde{L}\right)_{++} & = & tr(KHLH)\nonumber \\
 & = & tr(KH^{2}LH^{2})\nonumber \\
 & = & tr(HKH^{2}LH)\nonumber \\
 & = & \left(HKH\circ HLH\right)_{++}\nonumber \\
 & = & \left(\tilde{K}\circ\tilde{L}\right)_{++}.
\end{eqnarray}
This is no longer true in the three-variable case, where centering
of each matrix has a different meaning. Various hypotheses and their
corresponding V-statistics are summarized in Table \ref{tab: Vstat-summary}.
Note that the {}``composite'' hypotheses are obtained simply by
an appropriate centering of Gram matrices.

\section{\textmd{\normalsize \label{sec: counterexample}$\Delta_{L}P=0\nRightarrow(X,Y)\ci Z\,\vee\,(X,Z)\ci Y\,\vee\,(Y,Z)\ci X$.}}

Consider the following simple example with binary variables $X$,
$Y$, $Z$ with the $2\times2\times2$ probability table given in
Table \ref{tab: Example_of_Lancaster_vanishing}. It is readily checked
that all conditional covariances are equal, so $\Delta_{L}P=0$. It
is also clear, however, that neither variable is independent of the
other two. Therefore, a test for Lancaster interaction \emph{per se}
is not equivalent to testing for the possibility of any factorization
of the joint distribution, but our empirical results suggest that
it can nonetheless provide a useful surrogate. In other words, while
rejection of the null hypothesis $\Delta_{L}P=0$ is highly informative
and implies that interaction is present and \emph{no} non-trivial
factorization of the joint distribution is available, the acceptance
of the null hypothesis should be considered carefully and additional
methods to rule out interaction should be sought.

\begin{table}
\caption{\label{tab: Example_of_Lancaster_vanishing}{\footnotesize An example
of Lancaster interaction measure vanishing for the case where neither
variable is independent of the other two.}}

\centering{}{\small }%
\begin{tabular}{|c||c|}
\hline 
{\small $P(0,0,0)=0.2$} & {\small $P(0,0,1)=0.1$}\tabularnewline
\hline 
{\small $P(0,1,0)=0.1$} & {\small $P(0,1,1)=0.1$}\tabularnewline
\hline 
\hline 
{\small $P(1,0,0)=0.1$} & {\small $P(1,0,1)=0.1$}\tabularnewline
\hline 
{\small $P(1,1,0)=0.1$} & {\small $P(1,1,1)=0.2$}\tabularnewline
\hline 
\end{tabular}
\end{table}

\section{\label{sec: Permutation test}Permutation test}

A permutation test for total independence is easy to construct: it
suffices to compute the value of the statistic (either the Lancaster
statistic $\left\Vert \Delta_{L}\hat{P}\right\Vert _{k\otimes l\otimes m}^{2}$
or the total independence statistic $\left\Vert \Delta_{tot}\hat{P}\right\Vert _{k\otimes l\otimes m}^{2}$)
on $\left\{ \left(X^{(i)},Y^{(\sigma i)},Z^{(\tau i)}\right)\right\} _{i=1}^{n}$,
for randomly drawn independent permutations $\sigma,\tau\in S_{n}$
in order to obtain a sample from the null distribution. 

When testing for \emph{only one} of the hypotheses $(Y,Z)\ci X$,
$(X,Z)\ci Y$, or $(X,Y)\ci Z$, either with a Lancaster statistic
or with a standard two-variable kernel statistic, only one of the
samples should be permuted, e.g., if testing for $(Y,Z)\ci X$, statistics
should be computed on $\left\{ \left(X^{(\sigma i)},Y^{(i)},Z^{(i)}\right)\right\} _{i=1}^{n}$,
for $\sigma\in S_{n}$. However, when testing for the disjunction
of these hypotheses, i.e., for the existence of a nontrivial factorization
of the joint distribution, we are within a multiple hypothesis testing
framework (even though one may deal with a single test statistic,
as in the Lancaster case). To ensure that the required confidence
level $\alpha=0.05$ is reached for the factorization hypothesis,
in the experiments reported in Figure \ref{fig: VstructPlot}, the
Holm's sequentially rejective Bonferroni method \cite{Holm1979} is
used for both the two-variable based and for the Lancaster based factorization
tests. Namely, $p$-values are computed for each of the hypotheses
$(Y,Z)\ci X$, $(X,Z)\ci Y$, or $(X,Y)\ci Z$ using the permutation
test, and sorted in the ascending order $p_{(1)},p_{(2)},p_{(3)}$.
Hypotheses are then rejected sequentially if $p_{(l)}<\frac{\alpha}{4-l}$.
The factorization hypothesis is then rejected if and only if all three
hypotheses are rejected.

\section{Asymptotic behavior}

Using terminology from \cite{Sejdinovic2012a}, kernels $k$ and
$k'$ are said to be equivalent if they induce the same semimetric
on the domain, i.e., $k(x,x)+k(x',x')-2k(x,x')=k'(x,x)+k'(x',x')-2k'(x,x')$
$\forall x,x'$. It can be shown that the Lancaster statistic is invariant
to changing kernels within the kernel equivalence class, i.e., that
\begin{eqnarray*}
\left\Vert \Delta_{L}\hat{P}\right\Vert _{k\otimes l\otimes m}^{2} & = & \left\Vert \Delta_{L}\hat{P}\right\Vert _{k'\otimes l'\otimes m'}^{2},
\end{eqnarray*}
whenever $k,k'$, $l,l'$ and $m,m'$ are equivalent pairs. From here,
\begin{eqnarray*}
\left\Vert \Delta_{L}\hat{P}\right\Vert _{k\otimes l\otimes m}^{2} & = & \left\Vert \Delta_{L}\hat{P}\right\Vert _{\tilde{k}_{P_{X}}\otimes\tilde{l}_{P_{Y}}\otimes\tilde{m}_{P_{Z}}}^{2}.
\end{eqnarray*}
In Section \ref{sub:Proof-of-Proposition A2}, we were able to show
a similar expression but only for changing $k$ to its version $\tilde{k}_{\hat{P}_{X}}$
centered at the \emph{empirical marginal}. Now, under the assumption
of total independence, i.e., that $P_{XYZ}=P_{X}P_{Y}P_{Z}$, the
dominating term in $\left\Vert \Delta_{L}\hat{P}\right\Vert _{\tilde{k}_{P_{X}}\otimes\tilde{l}_{P_{Y}}\otimes\tilde{m}_{P_{Z}}}^{2}$
is $\left\Vert \hat{P}_{XYZ}\right\Vert _{\tilde{k}_{P_{X}}\otimes\tilde{l}_{P_{Y}}\otimes\tilde{m}_{P_{Z}}}^{2}$.
By standard arguments, under total independence, this converges in
distribution to a sum of independent chi-squared variables,
\begin{eqnarray}
n\left\Vert \hat{P}_{XYZ}\right\Vert _{\tilde{k}_{P_{X}}\otimes\tilde{l}_{P_{Y}}\otimes\tilde{m}_{P_{Z}}}^{2} & \rightsquigarrow & \sum_{a=1}^{\infty}\sum_{b=1}^{\infty}\sum_{c=1}^{\infty}\lambda_{a}\eta_{b}\theta_{c}N_{abc}^{2},\label{eq: 3chisquare}
\end{eqnarray}
where $\left\{ \lambda_{a}\right\} $, $\left\{ \eta_{b}\right\} $,
$\left\{ \theta_{c}\right\} $ are, respectively, eigenvalues of integral
operators associated to $\tilde{k}_{P_{X}}$, $\tilde{l}_{P_{Y}}$
and $\tilde{m}_{P_{Z}}$, and $N_{abc}\overset{i.i.d.}{\sim}\mathcal{N}(0,1)$.
Other terms in $\left\Vert \Delta_{L}\hat{P}\right\Vert _{\tilde{k}_{P_{X}}\otimes\tilde{l}_{P_{Y}}\otimes\tilde{m}_{P_{Z}}}^{2}$
can be shown to drop to zero at a faster rate, as in the two-variable
case. The resulting distribution of such a sum of chi-squares can,
in principle, be estimated using a Monte Carlo method, by computing
a number of eigenvalues of $\tilde{K}$, $\tilde{L}$ and $\tilde{M}$,
as in \cite{GreFukHarSri09,Zhang2011}. This is of little practical
value though, as it is in most cases simpler and faster to run a permutation
test, as we describe in Appendix \ref{sec: Permutation test}. On
the other hand, the above result quantifies the highest order of bias
of the V-statistic under total independence to be equal to $\frac{1}{n}\sum_{a=1}^{\infty}\lambda_{a}\sum_{b=1}^{\infty}\eta_{b}\sum_{c=1}^{\infty}\theta_{c}$,
which can be estimated as $\frac{1}{n^{4}}Tr(\tilde{K})Tr(\tilde{L})Tr(\tilde{M}).$
We emphasize that \eqref{eq: 3chisquare} refers to a \emph{null distribution
under total independence} - if say, the null holds because $(X,Y)\ci Z$,
but $X$ and $Y$ are dependent, one needs to instead consider a kernel
on $\mathcal{X}\times\mathcal{Y}$ centered at $P_{XY}$ and the eigenvalues
of its integral operator then replace $\left\{ \lambda_{a}\eta_{b}\right\} $
(triple sum becomes a double sum). This also implies that the bias
term needs to be corrected appropriately.

\section{\label{sec: matrix_algebra}Some useful basic matrix algebra}
\begin{lem}
Let $A$, $B$ be $n\times n$ matrices. The following results hold:
\begin{enumerate}
\item $\onet\one=n$
\item $[\one\onet]_{ij}=1,\;\forall i,j$, and thus $\left(\one\onet\right)_{++}=n^{2}$
\item \textup{$\left(I-\frac{1}{n}\one\onet\right)^{2}=I-\frac{1}{n}\one\onet.$}
\item $\left[A\one\right]_{i}=A_{i+}$, $\left[\onet A\right]_{j}=A_{+j}$
\item $\onet A\one=A_{++}$
\item $\left(A\one\onet\right)_{++}=\left(\one\onet A\right)_{++}=nA_{++}$
\item $\left(\alpha A+\beta B\right)_{++}=\alpha A_{++}+\beta B_{++}$
\item $\left(A\one\onet B\right)_{++}=A_{++}B_{++}$.
\end{enumerate}
\end{lem}
\begin{proof}
(3):
\begin{eqnarray*}
\left(I-\frac{1}{n}\one\onet\right)^{2} & = & I-\frac{2}{n}\one\onet+\frac{1}{n^{2}}\one\underset{n}{\underbrace{\onet\one}}\one^{\top}.
\end{eqnarray*}

(8): From (4), $\left[A\one\onet B\right]_{ij}=A_{i+}B_{+j}$, implying
\[
\left(A\one\onet B\right)_{++}=\sum_{i=1}^{n}A_{i+}\sum_{j=1}^{n}B_{+j}=A_{++}B_{++}.
\]

\end{proof}
Now, let $K$ be a symmetric matrix, and denote $H=I-\frac{1}{n}\one\onet$
(the centering matrix). Then:
\begin{eqnarray*}
HKH & = & \left(I-\frac{1}{n}\one\onet\right)K\left(I-\frac{1}{n}\one\onet\right)\\
 & = & K-\frac{1}{n}\left(K_{+}+K_{+}^{\top}\right)+\frac{1}{n^{2}}K_{++}\one\onet.
\end{eqnarray*}
Note that:
\begin{eqnarray*}
\left(HKH\right)_{++} & = & K_{++}-\frac{1}{n}\left(\left(K_{+}\right)_{++}+\left(K_{+}^{\top}\right)_{++}\right)+\frac{1}{n^{2}}K_{++}\left(\one\onet\right)_{++}\\
 & = & K_{++}-2K_{++}+K_{++}=0.
\end{eqnarray*}

\begin{lem}
The following results hold:
\begin{enumerate}
\item $A\circ\one\onet=\one\onet\circ A=A$
\item $\left(I\circ A\right)_{++}=tr(A)$
\item $\left(A\circ B\right)_{++}=tr(AB^{\top})$
\item \textup{For a symmetric matrix $K$ and any matrix $A$, $\left(A\circ K_{+}\right)_{++}=\left(AK\right)_{++}$,
$\left(A\circ K_{+}^{\top}\right)_{++}=\left(KA\right)_{++}$}
\item \textup{For symmetric matrices $K$, $L$, $\left(K_{+}\circ L_{+}\right)_{++}=\left(K_{+}^{\top}\circ L_{+}^{\top}\right)_{++}=n\left(KL\right)_{++}$}
\item \textup{For symmetric matrices $K$, $L$, $\left(K_{+}\circ L_{+}^{\top}\right)_{++}=\left(K_{+}^{\top}\circ L_{+}\right)_{++}=K_{++}L_{++}$}.
\end{enumerate}
\end{lem}
\begin{proof}
(4):$\left(A\circ K_{+}\right)_{++}=tr\left(AK\one\onet\right)=\left(AK\circ\one\onet\right)_{++}=\left(AK\right)_{++}.$
(5): $\left(K_{+}\circ L_{+}\right)_{++}=\left(K_{+}L\right)_{++}=\left(\one\onet KL\right)_{++}=n\left(KL\right)_{++}.$
\end{proof}
 
\begin{prop}
\label{prop: 2varcase-2}Denote $H=I-\frac{1}{n}\one\onet$. Then:

\textup{
\begin{eqnarray*}
\left(K\circ HLH\right)_{++} & = & \left(K\circ L\right)_{++}-\frac{2}{n}\left(KL\right)_{++}+\frac{1}{n^{2}}K_{++}L_{++}.
\end{eqnarray*}
}\end{prop}
\begin{proof}
Let $K$ and $L$ be symmetric matrices and consider $K\circ HLH$.
We obtain:
\begin{eqnarray*}
K\circ HLH & = & K\circ\left(L-\frac{1}{n}\left(L_{+}+L_{+}^{\top}\right)+\frac{1}{n^{2}}L_{++}\one\onet\right)\\
 & = & K\circ L-\frac{1}{n}\left(K\circ L_{+}+K\circ L_{+}^{\top}\right)+\frac{1}{n^{2}}L_{++}K,
\end{eqnarray*}
so that:
\begin{eqnarray*}
\left(K\circ HLH\right)_{++} & = & \left(K\circ L\right)_{++}-\frac{2}{n}\left(KL\right)_{++}+\frac{1}{n^{2}}K_{++}L_{++}.
\end{eqnarray*}
\end{proof}
\begin{cor}
\textup{$tr(HLH)=tr(L)-\frac{1}{n}L_{++}$}
\end{cor}


\footnotesize


\begin{thebibliography}{10}

\bibitem{GreBouSmoSch05}
A.~Gretton, O.~Bousquet, A.~Smola, and B.~Sch\"{o}lkopf.
\newblock Measuring statistical dependence with {H}ilbert-{S}chmidt norms.
\newblock In {\em ALT}, pages 63--78, 2005.

\bibitem{Szekely2007}
G.~Sz\'{e}kely, M.~Rizzo, and N.K. Bakirov.
\newblock Measuring and testing dependence by correlation of distances.
\newblock {\em Ann. Stat.}, 35(6):2769--2794, 2007.

\bibitem{Sejdinovic2012}
D.~Sejdinovic, A.~Gretton, B.~Sriperumbudur, and K.~Fukumizu.
\newblock Hypothesis testing using pairwise distances and associated kernels.
\newblock In {\em ICML}, 2012.

\bibitem{BacJor02}
F.~R. Bach and M.~I. Jordan.
\newblock Kernel independent component analysis.
\newblock {\em J. Mach. Learn. Res.}, 3:1--48, 2002.

\bibitem{FukBacGre07}
K.~Fukumizu, F.~Bach, and A.~Gretton.
\newblock Statistical consistency of kernel canonical correlation analysis.
\newblock {\em J. Mach. Learn. Res.}, 8:361--383, 2007.

\bibitem{DauNki98}
J.~Dauxois and G.~M. Nkiet.
\newblock Nonlinear canonical analysis and independence tests.
\newblock {\em Ann. Stat.}, 26(4):1254--1278, 1998.

\bibitem{pal10renyi}
D.~Pal, B.~Poczos, and Cs. Szepesvari.
\newblock Estimation of renyi entropy and mutual information based on
  generalized nearest-neighbor graphs.
\newblock In {\em NIPS 23}, 2010.

\bibitem{Kankainen95}
A.~Kankainen.
\newblock {\em Consistent Testing of Total Independence Based on the Empirical
  Characteristic Function}.
\newblock PhD thesis, {University of Jyv{\"a}skyl{\"a}}, 1995.

\bibitem{Bernstein46}
S.~Bernstein.
\newblock {\em The Theory of Probabilities}.
\newblock Gastehizdat Publishing House, Moscow, 1946.

\bibitem{Kayano2009}
M.~Kayano, I.~Takigawa, M.~Shiga, K.~Tsuda, and H.~Mamitsuka.
\newblock Efficiently finding genome-wide three-way gene interactions from
  transcript- and genotype-data.
\newblock {\em Bioinformatics}, 25(21):2735--2743, 2009.

\bibitem{MeiBuh06}
N.~Meinshausen and P.~Buhlmann.
\newblock High dimensional graphs and variable selection with the lasso.
\newblock {\em Ann. Stat.}, 34(3):1436--1462, 2006.

\bibitem{RavWaiRasYu11}
P.~Ravikumar, M.J. Wainwright, G.~Raskutti, and B.~Yu.
\newblock High-dimensional covariance estimation by minimizing
  $\ell_1$-penalized log-determinant divergence.
\newblock {\em Electron. J. Stat.}, 4:935--980, 2011.

\bibitem{Pearl01}
J.~Pearl.
\newblock {\em Causality: Models, Reasoning and Inference}.
\newblock Cambridge University Press, 2001.

\bibitem{spirtes:00}
P.~Spirtes, C.~Glymour, and R.~Scheines.
\newblock {\em Causation, Prediction, and Search}.
\newblock 2nd edition, 2000.

\bibitem{KalBuh07}
M.~Kalisch and P.~Buhlmann.
\newblock Estimating high-dimensional directed acyclic graphs with the {PC}
  algorithm.
\newblock {\em J. Mach. Learn. Res.}, 8:613--636, 2007.

\bibitem{SunJanSchFuk07}
X.~Sun, D.~Janzing, B.~Sch{\"o}lkopf, and K.~Fukumizu.
\newblock A kernel-based causal learning algorithm.
\newblock In {\em ICML}, pages 855--862, 2007.

\bibitem{TilGreSpi09}
R.~Tillman, A.~Gretton, and P.~Spirtes.
\newblock Nonlinear directed acyclic structure learning with weakly additive
  noise models.
\newblock In {\em NIPS 22}, 2009.

\bibitem{Zhang2011}
K.~Zhang, J.~Peters, D.~Janzing, and B.~Schoelkopf.
\newblock Kernel-based conditional independence test and application in causal
  discovery.
\newblock In {\em UAI}, pages 804--813, 2011.

\bibitem{GreFukTeoSonetal08}
A.~Gretton, K.~Fukumizu, C.-H. Teo, L.~Song, B.~Sch{\"o}lkopf, and A.~Smola.
\newblock A kernel statistical test of independence.
\newblock In {\em NIPS 20}, pages 585--592, Cambridge, MA, 2008. MIT Press.

\bibitem{FukGreSunSch08}
K.~Fukumizu, A.~Gretton, X.~Sun, and B.~Sch\"olkopf.
\newblock Kernel measures of conditional dependence.
\newblock In {\em NIPS 20}, pages 489--496, 2008.

\bibitem{Lancaster1969}
H.O. Lancaster.
\newblock {\em The Chi-Squared Distribution}.
\newblock Wiley, London, 1969.

\bibitem{Streitberg1990}
B.~Streitberg.
\newblock Lancaster interactions revisited.
\newblock {\em Ann. Stat.}, 18(4):1878--1885, 1990.

\bibitem{FukSriGreSch09}
K.~Fukumizu, B.~Sriperumbudur, A.~Gretton, and B.~Schoelkopf.
\newblock Characteristic kernels on groups and semigroups.
\newblock In {\em NIPS 21}, pages 473--480, 2009.

\bibitem{Kankainen1995}
A.~Kankainen.
\newblock {\em Consistent Testing of Total Independence Based on the Empirical
  Characteristic Function}.
\newblock PhD thesis, University of Jyv\"{a}skyl\"{a}, 1995.

\bibitem{BerTho04}
A.~Berlinet and C.~{Thomas-Agnan}.
\newblock {\em Reproducing Kernel Hilbert Spaces in Probability and
  Statistics}.
\newblock Kluwer, 2004.

\bibitem{Sejdinovic2012a}
D.~Sejdinovic, B.~Sriperumbudur, A.~Gretton, and K.~Fukumizu.
\newblock Equivalence of distance-based and {RKHS}-based statistics in
  hypothesis testing.
\newblock arXiv:1207.6076, 2012.

\bibitem{Sriperumbudur2011}
B.~Sriperumbudur, K.~Fukumizu, and G.~Lanckriet.
\newblock Universality, characteristic kernels and rkhs embedding of measures.
\newblock {\em J. Mach. Learn. Res.}, 12:2389--2410, 2011.

\bibitem{SriGreFukLanetal10}
B.~Sriperumbudur, A.~Gretton, K.~Fukumizu, G.~Lanckriet, and B.~Sch{\"o}lkopf.
\newblock Hilbert space embeddings and metrics on probability measures.
\newblock {\em J. Mach. Learn. Res.}, 11:1517--1561, 2010.

\bibitem{Gretton2012}
A.~Gretton, K.~Borgwardt, M.~Rasch, B.~Sch\"{o}lkopf, and A.~Smola.
\newblock A kernel two-sample test.
\newblock {\em J. Mach. Learn. Res.}, 13:723--773, 2012.

\bibitem{Szekely2004}
G.~Sz\'{e}kely and M.~Rizzo.
\newblock Testing for equal distributions in high dimension.
\newblock {\em InterStat}, (5), November 2004.

\bibitem{Baringhaus2004}
L.~Baringhaus and C.~Franz.
\newblock On a new multivariate two-sample test.
\newblock {\em J. Multivariate Anal.}, 88(1):190--206, 2004.

\bibitem{SzeRiz09}
G.~Sz\'ekely and M.~Rizzo.
\newblock Brownian distance covariance.
\newblock {\em Ann. Appl. Stat.}, 4(3):1233--1303, 2009.

\bibitem{Speed1983}
T.P. Speed.
\newblock Cumulants and partition lattices.
\newblock {\em Austral. J. Statist.}, 25:378--388, 1983.

\bibitem{Holm1979}
S.~Holm.
\newblock A simple sequentially rejective multiple test procedure.
\newblock {\em Scand. J. Statist.}, 6(2):65--70, 1979.

\bibitem{GreFukHarSri09}
A.~Gretton, K.~Fukumizu, Z.~Harchaoui, and B.~Sriperumbudur.
\newblock A fast, consistent kernel two-sample test.
\newblock In {\em NIPS 22}, Red Hook, NY, 2009. Curran Associates Inc.

\end{thebibliography}




\normalsize
\end{document}